%% file: main.new.tex
\documentclass[11pt]{article}
\usepackage{authblk}
\usepackage{amssymb}
\usepackage{amsthm}
\usepackage{mathtools}
\usepackage{amsfonts}
\usepackage{amsmath}

\title{Effective Continued Fraction Dimension versus Effective
	Hausdorff Dimension of Reals}  

\author[1]{Satyadev Nandakumar}
\author[1]{Akhil S}
\author[1]{Prateek Vishnoi}
\affil[1]{
	Department of Computer Science and Engineering\\
	Indian Institute of Technology Kanpur,
	Kanpur, Uttar Pradesh, India.
}
\affil[]{\{satyadev,akhis,pratvish\}@cse.iitk.ac.in}

\usepackage{amsmath,amsthm,amssymb}
\usepackage[margin=1in]{geometry}
\usepackage[colorlinks=true]{hyperref}

\include{commands}

\newcommand{\Rplus}{[0,\infty)}
\renewcommand{\F}{\mathcal{F}}

\begin{document}

\maketitle

\begin{abstract}
	We establish that constructive continued fraction dimension
	originally defined using $s$-gales \cite{Vishnoi22} is robust, but
	surprisingly, that the effective continued fraction dimension and
	effective (base-$b$) Hausdorff dimension of the same real can be
	unequal in general.
	
	We initially provide an equivalent characterization of continued
	fraction dimension using Kolmogorov complexity. In the process, we
	construct an optimal lower semi-computable $s$-gale for continued
	fractions. We also prove new bounds on the Lebesgue measure of
	continued fraction cylinders, which may be of independent interest.
	
	We apply these bounds to reveal an unexpected behavior of
	continued fraction dimension. It is known that feasible dimension is
	invariant with respect to base conversion \cite{HitchMayor13}. We
	also know that Martin-L\"of randomness and computable randomness are
	invariant not only with respect to base conversion, but also with
	respect to the continued fraction representation
	\cite{Vishnoi22}. In contrast, for any $0 < \varepsilon < 0.5$, we
	prove the existence of a real whose effective Hausdorff dimension is
	less than $\varepsilon$, but whose effective continued fraction
	dimension is greater than or equal to $0.5$. This phenomenon is
	related to the ``non-faithfulness'' of certain families of covers,
	investigated by Peres and Torbin \cite{PeresTorbin2013} and by
	Albeverio, Ivanenko, Lebid and Torbin \cite{Albeverio2020}. 
	
	We also establish that for any real, the constructive Hausdorff
	dimension is at most its effective continued fraction dimension.  
\end{abstract}

\section{Introduction}
\label{sec:introduction}
The concept of an individual random sequence, first defined by
Martin-L\"of using constructive measure \cite{MartinLoef1966}, is
well-established and mathematically robust - very different approaches
towards the definition identify precisely the same sequences as
random. These include Kolmogorov incompressibility (Levin
\cite{Levin1973a}, Chaitin \cite{Chaitin1975}) and unpredictability by
martingales \cite{Schnorr1971a}. While the theory of Martin-L\"of
randomness \emph{classifies} sequences into random and non-random, it
does not \emph{quantify} the information rate in a non-random
sequence. Lutz effectivized the classical notions of Hausdorff and
packing dimensions \cite{Lutz03}, surprisingly extending it to
individual infinite binary sequences \cite{Lutz2003}, yielding a
notion of information density in sequences. This definition also has
several equivalent definitions in terms of Kolmogorov compression
rates \cite{Mayordomo02}, unpredictability by $s$-gales, and using
covers \cite{Lutz03}, \cite{Lutz2003}. These definitions have led
to a rich variety of applications in various domains of computability
and complexity theory (see for example, Downey and Hirschfeldt
\cite{Downey10}, Nies \cite{Nies2009}).

Recently, settings more general than the Cantor space of infinite
binary (or in general, infinite sequences from a finite alphabet) have
been studied by Lutz and Mayordomo \cite{Lutz2008a}, and Mayordomo
\cite{Mayordomo2012}, \cite{Mayordomo2018}. Prominent among them is
Mayordomo's definition of effective Hausdorff dimension for a very
general class of metric spaces \cite{Mayordomo2012},
\cite{Mayordomo2018}. Nandakumar and Vishnoi \cite{Vishnoi22} and
Vishnoi \cite{Vishnoi2023normality} define the notion of effective
dimension of continued fractions, which involves a countably infinite
alphabet, and is thus a setting which cannot be studied using
Mayordomo's framework. This latter setting is interesting
topologically since the space of continued fractions is non-compact,
and interesting measure-theoretically since the natural shift
invariant measure, the Gauss measure, is a non-product measure.

Nandakumar and Vishnoi \cite{Vishnoi22} use the notion of an $s$-gale
on the space of continued fractions to define effective dimension.
Vishnoi \cite{Vishnoi2023normality} introduced the notion of
Kolmogorov complexity of finite continued fraction strings using a one
to one binary encoding. Vishnoi \cite{Vishnoi2023normality} also shows
that the notion of Kolmogorov complexity is invariant under computable
1-1 encodings, upto an additive constant.

In this work, we first establish the mathematical robustness of the
notion of effective dimension, by proving an equivalent
characterization using Kolmogorov complexity of continued fractions.
The characterization achieves the necessary equivalence by choosing a
binary encoding of continued fractions which has a compelling
geometric intuition, and then applying Mayordomo's characterization of
effective (binary) Hausdorff dimension using Kolmogorov complexity
\cite{Mayordomo02} . In the process, analogous to the notion of an
optimal constructive supergale on the Cantor space defined by Lutz
\cite{Lutz2003}, we provide the construction of a lower
semi-computable $s$-gale that is optimal for continued fractions. We
also prove new bounds on the Lebesgue measure of continued fraction
cylinders using the digits of the continued fraction expansion, a
result which may be of independent interest.

The topological and measure-theoretic intricacies involved in this
setting imply that some, but not all, ``natural'' properties of
randomness and dimension carry over from the binary setting. For
example, while Martin-L\"of and computable randomness are invariant
with respect to the conversion between the base-$b$ and continued
fraction expansion of the same real \cite{Nandakumar2008},
\cite{Vishnoi22}, Vandehey \cite{Vandehey2016} and Scheerer
\cite{Scheerer2017continued} show that other notions of randomness
like absolute normality and normality for continued fractions are not
identical.

Staiger \cite{Staiger02} showed that the
Kolmogorov complexity of a base $b$ expansion of a real
$\alpha, 0\le\alpha\le 1,$ is
independent of the chosen base $b$. Aligning with this,
Hitchcock and Mayordomo \cite{HitchMayor13} establish that feasible
dimension of a real is the same when converting between one base to
another. 
Hitherto, it was unknown whether effective dimension is
invariant with respect to conversion between base-$b$ and continued
fraction representations. Since we can convert between the
representations efficiently, it is possible that these are equal. We
show this is true in one direction, that the effective base $b$
dimension is a lower bound for effective continued fraction dimension.

However, using the technique of diagonalization against the optimal
lower semicomputable continued fraction $s$-gale and using set
covering techniques used in recent works by Peres and Torbin
\cite{PeresTorbin2013}, Albeverio, Ivanenko, Lebid and Torbin
\cite{Albeverio2020} and Albeverio, Kondratiev, Nikiforov and Torbin
\cite{Albeverio2017} to show the ``non-faithfulness'' of certain
families of covers, we show that the reverse direction does not hold,
in general. We prove the following result: for every $0 < \varepsilon
< 0.5$, there is a real whose effective (binary) Hausdorff dimension
is less than $\varepsilon$ while its effective continued fraction
dimension is at least 0.5. By the result of Hitchcock and Mayordomo
\cite{HitchMayor13}, this also implies that the effective base-$b$
dimension of this real is less than $\varepsilon$ in every base-$b$,
$b \ge 2$. Thus, surprisingly, there is a sharp gap between the
effective (base-$b$) dimension of a real and its effective continued
fraction dimension, highlighting another significant difference in
this setting.

\section{Preliminaries}

We denote the binary alphabet by $\Sigma$. The set of strings of a
particular length $n$ is denoted $\Sigma^n$. The set of all finite
binary strings is denoted $\Sigma^*$ and infinite binary sequences is
denoted $\Sigma^\infty$. For a binary string $v \in \Sigma^n \setminus
\{0^n \cup 1^n\}$, $v-1$ denotes the string occurring just before $v$
lexicographically, and $v+1$ the string occurring just after $v$
lexicographically. We use $\N$ to denote the set of positive integers.
The set of finite continued fractions is denoted $\N^*$ and the set of
all infinite continued fractions, as $\N^\infty$. 

We adopt the
notation $[a_1, a_2, \dots]$ for the continued fraction
\begin{align*}
	\cfrac{1}{a_1+\cfrac{1}{a_2+\cdots}}
\end{align*}
and similarly, $[a_1, a_2, \dots, a_n]$ for finite continued
fractions.

If a finite binary string $x$ is a prefix of a finite string $z$ or an
infinite binary sequence $Z$, then we denote this by $x \sqsubseteq z$
or $x \sqsubseteq Z$ respectively.  If $x$ is a proper
prefix of a finite string $z$, we denote it by $x \sqsubset z$. We adopt the same notation for
denoting that a finite continued fraction $v$ is a prefix of another
continued fraction. For a $v \in \N^*$, the \emph{cylinder set} of
$v$, denoted $C_v$, is defined by $C_v=\{ Y \in \N^\infty \mid v
\sqsubset Y\}$. For a $w \in \Sigma^*$, $C_w$ is defined similarly.
For a continued fraction string $v = [a_1 \dots a_n]$, $P(v)$ denotes
the string $[a_1, \dots a_{n-1}]$. $\lambda$ denotes the empty string and we define $P(\lambda) = \lambda$. 

For $v \in \N^*$, $\mu(v)$ refers to the Lebesgue measure of the
continued fraction cylinder $C_v$. $\gamma(v)$ refers to the Gauss
measure of the continued fraction cylinder $C_v$, defined by
$\gamma(v) = \int_{C_v} \frac{1}{1+x}~dx$. We use the same notation
for a binary cylinder $w \in \Sigma^*$. It is well-known that the
Gauss measure is absolutely continuous with respect to the Lebesgue
measure, and is invariant with respect to the left-shift
transformation on continued fractions (see for example,
\cite{Dajani2002}, or \cite{Einsiedler2010}). Wherever there is no scope for confusion, for a $v \in \N^*$, we use $\mu(v)$ and $\gamma(v)$ to represent $\mu(C_v)$ and $\gamma(C_v)$ respectively. The same holds for a $v \in \Sigma^*$. We also use the notation $\mu^s(v)$ and $\gamma^s(v)$ to denote $(\mu(v))^s$ and $(\gamma(v))^s$ respectively. For a continued fraction
string $v = [a_1, \dots, a_n]$, we call $n$ as the rank of $v$, and we denote it using $rank(v)$.
$[v,i]$ denotes the continued fraction $[a_1 , \dots a_n , i]$. For an infinite continued
fraction string $Y = [a_1, a_2, \dots]$, $Y \upharpoonleft n$ denotes
the continued fraction string corresponding to the first n entries of
$Y$, that is $Y \upharpoonleft n = [a_1 , a_2 \dots a_n]$. For $k \in
\N$, $\N^{\leq k}$ refers to the set of continued fraction strings
having rank less than or equal to $k$. All logarithms in the work have
base 2, unless specified otherwise. For any sets $A$ and $B$, $A\Delta B$ denotes the symmetric set
difference operator, defined by $(A \setminus B) \cup (B \setminus
A)$. In this work, for ease of notation, $Y \in \N^*$ denotes an
infinite continued fraction and $X \in \Sigma^\infty$ denotes an infinite
binary sequence.

\subsection{Constructive dimension of binary sequences}

Lutz \cite{Lutz2003} defines the notion of effective (equivalently,
constructive) dimension of an individual infinite binary sequence
using the notion of the success of $s$-gales.
\begin{definition}[Lutz \cite{Lutz2003}]
	For $s \in \Rplus$, a binary $s$-gale is a function $d: \Sigma^{*}
	\to \Rplus$ such that $d(\lambda) < \infty$ and for all $w \in
	\Sigma^*$, $ d(w)[\mu (C_w)]^s = \sum_{i \in \{0,1\}} d(wi)
	[\mu(C_{wi})]^s.$
	
	The \emph{success set} of $d$ is $S^\infty (d) = \left\{ X \in
	\N^\infty \mid \limsup \limits_{n\to\infty} d(X \restr n) =
	\infty\right\}$.
	
	For $\mathcal{F} \subseteq [0,1]$, $\mathcal{G}(\mathcal{F})$
	denotes the set of all $s \in \Rplus$ such that there exists a lower
	semicomputable binary $s$-gale $d$ with $\mathcal{F} \subseteq S^\infty (d)$.
	
	The \emph{constructive dimension} or \emph{effective Hausdorff
		dimension} of $\mathcal{F} \subseteq [0,1]$ is $\cdim(\mathcal{F}) =
	\inf \mathcal{G}(\mathcal{F})$ and the constructive dimension of a
	sequence $X \in \Sigma^\infty$ is $\cdim(X) = \cdim(\{X\})$.
\end{definition}

\section{Effective Continued Fraction Dimension using $s$-gales}
\label{sec:sgaleDimension}

Nandakumar and Vishnoi \cite{Vishnoi22} formulate the notion of
effective dimension of continued fractions using the notion of lower
semicomputable continued fraction $s$-gales. Whereas a binary $s$-gale
bets on the digits of the binary expansion of a number, a continued
fraction $s$-gales places bets on the digits of its continued fraction
expansion.

\begin{definition}[Nandakumar, Vishnoi \cite{Vishnoi22}]
	For $s \in \Rplus$, a continued fraction $s$-gale is a function $d:
	\mathbb{N}^{*} \to \Rplus$ such that $d(\lambda) < \infty$ and for all
	$w \in \mathbb{N}^*$, the following holds.
	\begin{align*}
		d(w)[\gamma (C_w)]^s = \sum_{i \in \mathbb{N}} d(wi)
		[\gamma(C_{wi})]^s.
	\end{align*}
	The \emph{success set} of $d$ is $S^\infty (d) = \left\{ Y \in
	\N^\infty \mid \limsup \limits_{n\to\infty} d(Y \restr n) =
	\infty\right\}$.
\end{definition}

In this paper, we deal with the notion of effective or equivalently,
constructive dimension. In order to effectivize the notion
of $s$-gales, we require them to be \emph{lower semicomputable}.

\begin{definition} 
	A function $d : \mathbb{N^*} \longrightarrow \Rplus$ is called
	\emph{lower semicomputable} if there exists a total computable
	function $\hat{d} : \mathbb{N^*} \times \mathbb{N} \longrightarrow
	\mathbb{Q} \cap \Rplus $ such that the following two conditions
	hold.
	\begin{itemize}
		\item \textbf{Monotonicity} : For all $w \in \mathbb{N^*}$ and
		for all $n \in \mathbb{N}$, we have $ \hat{d}(w,n) \leq
		\hat{d}(w,n+1) \leq d(w)$.
		\item \textbf{Convergence} : For all $w \in \mathbb{N^*}$,
		$\lim\limits_{n\to\infty} \hat{d}(w,n) = d(w)$.
	\end{itemize}
\end{definition}

For $\mathcal{F} \subseteq [0,1]$, $\mathcal{G}_{CF}(\mathcal{F})$
denotes the set of all $s \in \Rplus$ such that there exists a lower
semicomputable continued fraction $s$-gale $d$ with $\mathcal{F}
\subseteq S^\infty (d)$.

\begin{definition}[Nandakumar, Vishnoi \cite{Vishnoi22}]
	The \emph{effective continued fraction dimension} of $\mathcal{F}
	\subseteq [0,1]$ is
	\begin{align*}
		\cdim_{CF}(\mathcal{F}) = \inf \mathcal{G}_{CF}(\mathcal{F}).
	\end{align*}
	The \emph{effective continued fraction dimension} of a sequence $Y \in
	\N^\infty$ is defined by $\cdim_{CF}(\{Y\})$, the effective continued fraction dimension
	of the singleton set containing $Y$.
\end{definition}

\subsection{Conversion of binary $s$-gales into continued fraction
	$s$-gales} In this subsection, from a continued fraction $s'$-gale
$d:\N^* \to \Rplus$, for any $s> s'$, we construct a binary $s$-gale
$h:\Sigma^* \to \Rplus$ which succeeds on all the reals on which $d$
succeeds. The construction proceeds in multiple steps. We first
mention some technical lemmas which we use in the proof.

The following lemma is an easy consequence of the fact that the Gauss
measure is absolutely continuous with respect to the Lebesgue measure
(see for example, Nandakumar and Vishnoi \cite{Vishnoi22}).

\begin{lemma}\label{lem:lebesgue_gauss} 
	For any interval $B \subseteq (0,1)$, we have
	$$\frac{1}{2 \ln 2} \mu(B) \le \gamma(B) \le \frac{1}{\ln 2} \mu(B).$$
\end{lemma}

In the construction that follows, we formulate betting strategies on
binary cylinders based on continued fraction cylinders. In order to do
this conversion, we require the following bounds on the relationships
between the lengths of continued fraction cylinders and binary
cylinders.

\begin{lemma} [Nandakumar, Vishnoi
	\cite{Vishnoi22}] \label{lem:vishnoiMeasure} For any $0 \leq a < b
	\leq 1$, let $\left[\frac{m}{2^k}, 
	\frac{m+1}{2^k}\right)$, where $0 \le m \le 2^k-1$ be one of the largest dyadic intervals which
	is a subset of $[a,b)$, then $\frac{1}{2^k} \geq \frac{1}{4} (b-a)$.
\end{lemma}

\begin{lemma} [Falconer \cite{Falc03}]
	\label{lem:binary_covering_cf} 
	For any $0 \le a < b \le 1$, let $\left[\frac{m}{2^k},
	\frac{m+1}{2^k}\right)$ , $\left[\frac{m+1}{2^k},
	\frac{m+2}{2^k}\right)$, where $0 \le m \leq 2^k - 2$, be the
	smallest consecutive dyadic intervals whose union covers
	$[a,b)$. Then $\frac{1}{2^k} \le 2(b-a)$.
\end{lemma}

\begin{proof}	
	Let $j = \lfloor - \log_2(b-a) \rfloor $. Since $(b-a)/2 < 2^{-j}$, it
	follows that at most two dyadic rationals of the form $m/2^j$, $0 \le
	m < 2^j$ is in $(a,b)$. Thus, three dyadic intervals of length
	$\frac{1}{2^j}$ cover the interval $[a,b]$. Hence two dyadic intervals
	of length $\frac{2}{ 2^j}$ cover $[a,b]$. Since $2^{-j} \le (b-a)$, it
	follows that $\frac{2}{ 2^j} \leq 2(b-a)$.
\end{proof}

The following lemma is a generalization of the Kolmogorov
inequality for continued fraction martingales (Vishnoi
\cite{VishnoiThesis}) to $s$-gales. The lemma states that an equality holds in the case of decompositions using prefix-free subcylinder sets upto a finite depth.

\begin{lemma} \label{lem:kolmogorovEquality}
	Let $d : \N^* \rightarrow \Rplus$ be a continued fraction
	$s$-gale. Let $v \in \N^*$ and for some $k \in \N$, let $A$ be a
	prefix free set of elements in $\N^{\leq k}$ such that $\cup_{w \in A}
	C_w = C_v$. Then, we have $d(v) \gamma^s(v) = \sum_{w\in A} d(w)
	\gamma^s(w)$.
\end{lemma}

\begin{proof} 
	We prove this result by induction on $k$, the maximum
	rank of an element in $A$. We first observe that for any $v \in \N^*$, if $C_v = \cup_{w \in A}
	C_w$, then for all $w \in A$, $v \sqsubseteq w $. Therefore, we have that $k \geq rank(v)$.
	
	Let's consider the case when $k = rank(v)$. In this case, the only possibility is
	that $A = \{v\}$ and therefore the lemma
	holds trivially.
	
	
	Assume that the lemma holds for any $k \in \N$ such that $k \geq rank(v)$.	
	Let $A$ be a prefix free set of elements in $\N^{\leq k+1}$ such that $\cup_{w \in A} C_w = C_v$.
	
	Consider the sets, $P = \{ w \in A ~|~
	rank(w) = k+1\}$ and $P'= \{ u \in \N^k ~|~ [u,i] \in P \text{ for some } i \in \N \}$.
	
	Now since $ \cup_{w \in A} ~C_w = C_v$ and $A$ is prefix free, we have that for all $u \in P'$, and for all $i \in \N$, $[u,i] \in P$. 
	
	Since $d$ is a continued fraction $s$-gale, it follows that,
	
	$$\sum_{u \in P'}d(u)\gamma^s(u) = \sum_{w \in P} d(w)
	\gamma^s(w)$$ 
	
	Also, note that since for any $u \in \mathbb{N^*}$, $C_u =
	\cup_{i \in \N}~C_{[u,i]}$, we have
	$$\cup_{u \in P'} ~C_u = \cup_{w \in P}~C_w.$$
	
	Construct the $A'= P' \cup (A - P)$. We have that the elements in $A'$ are prefix free and  $\cup_{w \in A'}
	C_w = C_v$. Since the maximum rank of an element in $A'$ is $k$, we have that,
	
	$$d(v) \gamma^s(v) = \sum_{w\in A'} d(w) \gamma^s(w)$$
	
	Since $P' \cap (A - P) = \phi$, it follows that
	
	$$d(v)\gamma^s(v) = \sum_{w \in A - P}
	d(w) \gamma^s(w) + \sum_{\substack{u \in P'}} d(u) \gamma^s(u).$$ 
	
	Using $\sum_{u \in P'}d(u)\gamma^s(u) = \sum_{w \in P} d(w)
	\gamma^s(w)$, we have
	
	$$d(v)\gamma^s(v) = \sum_{w \in A - P}
	d(w) \gamma^s(w) + \sum_{\substack{w \in P}} d(w) \gamma^s(w).$$ 
	
	From which we have
	
	$$d(v)\gamma^s(v) = \sum_{w \in A} d(w) \gamma^s(w) .$$
\end{proof}

In the construction of a binary $s$-gale from continued fraction gales, The first step is the following decomposition of 
a binary cylinder into a set of prefix free continued fraction cylinders.

\begin{lemma}[Vishnoi \cite{VishnoiThesis}]
	\label{lem:binarytocfDivision}
	For every $w \in \Sigma^*$, there exists a set $I(w) \subseteq \N^*$
	and a constant $k \in \N$ such that, 
	\begin{enumerate}
		\item $y \in \N^{\leq k}$ for every $y \in I(w)$.
		\item $(\cup_{y \in I(w)} C_y ) \Delta C_w \subseteq \{\inf(C_w),
		\sup(C_w)\}$ 
		\item $I(w0) \cup I(w1) = I(w)$
		\item $I(w0) \cap I(w1) = \phi$ 
	\end{enumerate}
\end{lemma}

Moreover, given $w\in\Sigma^*$, Vishnoi \cite{VishnoiThesis} gives a
division algorithm to compute $I(w)$. It is also clear from the
division algorithm that for all $w \in \Sigma^*$, there exists a $u
\in I(w)$ such that for all $v \in (I(w0) \cup I(w1)) \setminus I(w)$,
we have $u \sqsubset v$. This $u \in I(w)$ is the continued fraction
cylinder for which the mid point of $w$, $m(w)$ is an interior point
in $C_u$ and therefore gets divided.

From a continued fraction martingale, Vishnoi \cite{VishnoiThesis} uses the decomposition $I(w)$ to construct a binary martingale that places the same aggregate bets on an interval. 
We generalize this construction to the setting of
$s$-gales. Given a continued fraction $s'$-gale $d:\N^* \to
\lcro{0}{\infty}$, using the decomposition $I(w)$, we construct a
binary $s'$-gale $H_d$ from $d$.

\begin{definition} \label{def:proportionalSgale}
	Given any continued fraction $s'$-gale $d: \N^* \to \lcro{0}{\infty}$,
	define the \emph{Proportional binary $s'$-gale} of $d$, $H_d :\Sigma^*
	\to \lcro{0}{\infty}$ as follows:
	$$
	H_d(w) = \sum_{{y \in I(w)}}d(y) \left(\frac{
		\gamma(y)}{\mu(w)}\right)^{s'}.
	$$
\end{definition}

For a $w \in \Sigma^*$, let $I'(w) = I(w0) \cup I(w1)$. Then we have,
\begin{align*}
	H_d(w0) + H_d(w1) = 2^{s'}  \sum_{y \in I'(w) } d(y)
	\Big(\frac{ \gamma(y)}{\mu(w)}\Big)^{s'}.
\end{align*}
Let $u \in I(w)$ such that for all $v \in
I'(w) \setminus I(w)$, $u \sqsubset v$. Hence, by Lemma
\ref{lem:kolmogorovEquality}, it follows that  $\sum_{y \in I'(w) } d(y) \gamma^{s'}(y) = \sum_{y \in I(w) } d(y) \gamma^{s'}(y)$. Therefore, we have 
$H_d(w0) + H_d(w1) = 2^{s'}
H_d(w)$, so $H_d$ is an $s'$-gale. Also as $\gamma(\lambda) =
1$, we have that $H_d(\lambda) = d(\lambda)$.

As $I(w)$ is computably enumerable, it follows that $H_d$ is lower
semicomputable if $d$ is lower semicomputable.

The construction by Vishnoi \cite{VishnoiThesis} proceeds using the
\emph{savings-account trick} for martingales. In the setting of
$s$-gales, however, the concept of a savings account does not work
directly. Therefore, we require additional constructions in this
setting.

Using ideas from the construction given in Lemma 3.1 in Hitchcock and
Mayordomo \cite{HitchMayor13}, we construct a ``smoothed'' $s$-gale
$H_h : \Sigma^* \to \lcro{0}{\infty}$ from the proportional $s'$-gale constructed in Definition \ref{def:proportionalSgale}.

\begin{definition} \label{def:smoothedSgale}
	For a $w \in \Sigma^*$, and an $n > |w|$, we define 
	\begin{align*}
		F_n(w) &= \{u \in \{0^n \cup 1^n\} \mid w \prefix{u}\} \;\cup\;
		\{u \in \Sigma^n \setminus \{0^n \cup 1^n\} \mid w \prefix u+1 \text{
			and } w \prefix u-1\},\\ H_n(w) &= \{u \in \Sigma^n \mid w
		\prefix{u} \text{ or } w \prefix u + 1 \text{ or } w \prefix u - 1\} \setminus
		F_n.
	\end{align*}
\end{definition}


\begin{definition}
	Given an $s'$-gale $h: \Sigma^* \to \lcro{0}{\infty}$, for any $s>s'$
	and for each $n \in \N$, define:
	\begin{align*}
		h_n(w) = \begin{cases} \;\; 2^{s |w|} \; \bigg( \;
			\sum\limits_{u \in H_n(w)} \frac{1}{2} \; h(u) \;+\;
			\sum\limits_{u \in F_n(w)} h(u) \bigg) & \text {if }
			|w| < n\\\\ \;\; {2^{(s-1) (|w| - n + 1)}} \;\;
			h_n(w[0\dots n-2]) & \text {otherwise.}\end{cases}
	\end{align*}
	
	Define $S_h : \Sigma^* \to \lcro{0}{\infty}$ by
	$$ S_h(w) = \sum\limits_{n=0}^\infty 2^{-sn} \; h_n(w).$$
	
	We call $S_h$ as the \emph{smoothed $s$-gale of $h$}.
\end{definition}

Consider a string $w \in \Sigma^n$ other than $0^n$ and $1^n$. In  $h_n$, a factor of half the capital of $w$ gets assigned to it's immediate parent $w'$. The
other half is assigned to the neighbor of $w'$ to which $w$ is
adjacent to. 

It is straightforward to verify that each $h_n$ is an $s$-gale. $S_h$
is a combination of $s$-gales, and hence is a valid $s$-gale.  Note
that $h_n(\lambda) = \sum_{u \in \Sigma^n}h(u) = 2^{s'n}$.  Therefore
as $s>s'$, $S_h(\lambda) = \sum_{n \in \N} 2^{(s'-s)n}$ is finite.  If
$h$ is lower semicomputable, it follows that $S_h$ is lower
semicomputable.

Combining the constructions given in the section, for any $s > s' $,
we show the construction of a binary $s$-gale from a continued
fraction $s'$-gale, satisfying certain bounds on the capital acquired.

This construction helps to establish a lower bound on effective
continued fraction dimension using effective binary dimension. It is also
central in formulating a Kolmogorov complexity characterization for
continued fraction dimension.
\begin{lemma} \label{lem:untitledUsefuLemma}
	For $s' \in (0,\infty)$, let $d: \N^* \to \lcro{0}{\infty}$ be a
	continued fraction $s'$- gale. Then, for any $s > s'$, there exists a
	binary $s$-gale $h:\Sigma^* \to \lcro{0}{\infty}$ such that for any $v
	\in \N^*$ and for any $b \in \Sigma^*$ such that $C_b \cap C_v \neq
	\phi$ and $ \frac{1}{16} \mu(v) \leq \mu(b) \leq 2\mu(v)$, we have
	$$h(b) \geq c_s d(v),$$ where $c_s$ is a constant that depends on $s$.
	Moreover, if $d$ is lower semicomputable, then $h$ is lower
	semicomputable.
\end{lemma}

\begin{proof}
	Given an $s'$-gale $d: \N^* \to	\lcro{0}{\infty}$, 
	let $h' = H_d$, the proportional binary $s'$-gale of $d$ given in Definition \ref{def:proportionalSgale}.
	
	For a $v \in \N^*$, consider the smallest $w_1,w_2 \in \Sigma^*$ such that $C_v \subseteq C_{w_1} \cup C_{w_2}$. 
	Also we can see that there exists a $S \subseteq I(w_1) \cup I(w_2)$ such that $C_v = \cup_{u \in S} C_u$. 
	
	Therefore from Lemma \ref{lem:kolmogorovEquality}, we get $h'(w_1) + h'(w_2) \geq d(v) \frac{\gamma^{s'}(v)}{\mu^{s'}(w_1)}$. 
	From Lemma \ref{lem:binary_covering_cf}, we get that $\mu(w_1) = \mu(w_2) \leq 2 \mu(v)$. Also from Lemma \ref{lem:lebesgue_gauss}, we have that $\gamma(v) \leq (ln2)^{-1} \mu(v)$. Therefore, $h'(w_1) + h'(w_2) \geq (2 ln2)^{-s'} d(v)$. Now for any $s>s'$, we have that $h'(w_1) + h'(w_2) \geq c_1 . d(v)$, where $c_1 = 1/(2ln2)^s$.
	
	Now for any $s>s'$, consider the smoothed $s$-gale $h = S_{H_d}$ of the $s'$-gale $H_d$ given in Definition \ref{def:smoothedSgale}. 
	
	Let $|w_1| = n$ , and let $W_1 = P(P(w_1))$ be the parent cylinder of parent of $w_1$. Similarly let $W_2 = P(P(w_2))$. We see that for any $W \in \{W_1,W_2\}$, $h_n(W) \geq 2^{s(n-2)} \frac{h'(w_1)+h'(w_2)}{2} \geq c_2.2^{sn}.d(v)$, where $c_2 = 2^{-(2s+1)} c_1$.
	
	Take any any $b \in \Sigma^*$ such that $C_b \cap C_v \neq \phi$ and $ 2.\mu(v) \geq \mu(b) \geq \frac{1}{16} \mu(v)$. Since $\mu(b) \leq 2\mu(v)$ and $\mu(v) \leq 2.\mu(w_1)$, it follows that for some  $W \in \{W_1,W_2\}$, $W \prefix b$.
	Also since $\mu(b) \geq \frac{1}{16} \mu(v)$, we have that, $\mu(b) \geq \frac{1}{32} \mu(w_1)$.
	
	Therefore, we have that $h_n(b) \geq 2^{5(s-1)}h_n(W) \geq c_3.2^{sn}.d(v)$, where $c_3 = c_2.2^{5(s-1)}$. 
	
	Since $h(b) \geq 2^{-sn} h_n(b)$, we have that $h(b) \geq c_3. d(v)$. 
\end{proof}

\section{Kolmogorov Complexity characterization of Continued Fraction Dimension}

Mayordomo \cite{Mayordomo02} extended the result by Lutz
\cite{Lutz2000} to show that effective dimension of a binary sequence $X \in
\Sigma^\infty$ can be characterized in terms of the Kolmogorov
complexity of the finite prefixes of $X$.

\begin{theorem} [Mayordomo \cite{Mayordomo02} and Lutz
	\cite{Lutz2000}] \label{thm:cdimAndliminf} 
	For every $X \in \Sigma^\infty$,$$\cdim(X) =
	\liminf\limits_{n \rightarrow \infty} \frac{K(X \restr n)}{n}.$$ 
\end{theorem}

We provide a similar characterization for effective continued fraction
dimension. To obtain the Kolmogorov complexity of a continued fraction
string, we use the Kolmogorov complexity of one of its binary
encodings. 

The idea of encoding a finite continued fraction using a 1-1 binary
encoding is present in Vishnoi \cite{Vishnoi2023normality}. The author
presents an invariance theorem stating that every computable binary
1-1 encoding of continued fractions defines the same Kolmogorov
complexity, up to an additive constant. Hence in this work, we use a
new binary encoding to define Kolmogorov complexity of continued
fractions, which helps us establish the characterization of effective
dimension of continued fractions in a fairly simple manner while
having intuitive geometric meaning.

\begin{definition} [Many-one binary encoding]
	For a continued fraction string $v \in \N^*$, let $b_v$ be the
	leftmost maximal binary cylinder which is enclosed by $C_v$. We define
	$E(v) = b_v$.
\end{definition}

\begin{lemma} \label{lem:cfEncoingMax3}
	For any $b \in \Sigma^*$, there exists at most three $v \in \N^*$ such
	that $E(v) = b$.
\end{lemma}

\begin{proof}
	For $b \in \{0,1\}^*$ assume there  exists distinct $v_1, v_2, v_3, v_4 \in \N^*$ such that $E(v_1) = E(v_2) = E(v_3) = E(v_4) = b$. Since $b$ is enclosed by all of these continued fraction cylinders, it follows that these continued fraction strings are extensions of each other. Assume without loss of generality that $v_1 \sqsubset v_2 \sqsubset v_3 \sqsubset v_4$. We can also assume that these cylinders are one length extensions of each other. From Lemma \ref{lem:Kraaikamp}, using the fact that $s_n >0$ and $i \geq 1$, it follows that for any $v \in \N^*$ and $i \in \N$, $\frac{\mu[v,i]}{\mu[v]} \leq 1/2$ . Therefore we have, $\mu(v_4) \leq \frac{1}{8} \mu(v_1)$. But from Lemma \ref{lem:vishnoiMeasure}, we get $\mu(b) \geq \frac{1}{4}\mu(v_1)$. Combining these two, we get that $\mu(b) \geq 2 \mu(v_4)$, which leads to a contradiction, since $b$ is enclosed by $v_4$. 
\end{proof}

Therefore for any $b \in \Sigma^*$, at most three continued fraction
cylinders, say $[v]$, $[v,i]$ and $[v,i,j]$ get mapped to $b$.
Therefore we pad two additional bits of information to $E(v)$, say
$b_1(v).b_2(v)$ to identify the continued fraction cylinder that
$E(v)$ corresponds to.

\begin{definition}[One-one binary encoding]
	
	For $v \in \N^*$, let $\mathcal{E}(v) = E(v).b_1(v).b_2(v)$. This
	forms a one to one binary encoding of $v$.
	
\end{definition}

We define Kolmogorov complexity of continued fraction string $v \in
\N^*$ as the Kolmogorov complexity of $\mathcal{E}(v)$.

\begin{definition} [Kolmogorov complexity of continued fraction strings]
	For any $v \in \N^*$, define $K_{\mathcal{E}}(v) = K(\mathcal{E}(v))$. 
\end{definition}

{\bf Notation.} By the invariance theorem of Vishnoi
\cite{Vishnoi2023normality}, for any $v \in \Sigma^*$,
$K_{\mathcal{E}}$ is at most an additive constant more than the
complexity of $v$ as defined in \cite{Vishnoi2023normality}. Hence, we
drop the suffix and denote the above complexity as $K(v)$.

In the proof of Theorem \ref{thm:cdimAndliminf}, Mayordomo
\cite{Mayordomo02} provides the construction of an $s$-gale that
succeeds on all $X$ for which $s > s' > \liminf_{n \rightarrow \infty}
\frac{K(X \restr n)}{n}$. We extend the construction to the setting of
continued fractions.

Additionally, we take a convex combination of gales to remove the
dependence of the $s$-gale on the parameter $s'$. Due to this, we
obtain the notion of an optimal lower semicomputable continued
fraction $s$-gale. This notion is crucial in the proofs we use in the
upcoming sections.

\begin{definition}\label{def:optSgaleD}
	Given  $0 < s' < s \leq 1$ let
	\begin{align*}
		G_{s'} = \{w \in \N^* \mid K(w) \leq - s' \log(\mu(w))\}.
	\end{align*}
	
	Consider the following function $d_{s'} : \N^* \rightarrow \Rplus$
	defined by
	\begin{align*}
		d_{s'}(v) =\frac{1}{\gamma^{s}(v)} 
		\left( 
		\sum_{w \in G_{s'} ; v \prefix w} \gamma^{s'}(w) 
		+ \sum_{w \in G_{s'} ; w \sqsubset v} \gamma^{s'}(w)
		\frac{\gamma(v)}{\gamma(w)}\right).
	\end{align*}
	
	Now for each $i \in \N$, let $s_{i} = s(1- 2^{-i})$. Finally, define
	$d^*:\N^* \to \Rplus$ by
	\begin{align*}
		d^*(v) = \sum\limits_{i = 1}^\infty2^{-i} d_{s_i}(v).
	\end{align*}
\end{definition}

We now go on to show that the function $d^*$ given in Definition
\ref{def:optSgaleD} is a lower semicomputable $s$-gale. Additionally,
it succeeds on all continued fraction sequences $Y$ for which the
Kolmogorov complexity of its prefixes, $K(Y\restr n)$ dips below $s
\times -\log(\mu(Y \upharpoonleft n))$ infinitely often.

\begin{lemma} \label{lem:dimlessthanliminf}
	For any $0 < s \leq 1$, there exists a lower semicomputable continued
	fraction $s$-gale $d^* : \N^* \rightarrow \lcro{0}{\infty}$ that succeeds on  all $Y \in \N^\infty$ such that $\liminf\limits_{n
		\rightarrow \infty} \frac{K(Y \upharpoonleft n)}{- \log(\mu(Y
		\upharpoonleft n))}< s$.
\end{lemma}

\begin{proof}
	Consider any $s'< s \leq 1$. Let $G_{s'} = \{w \in \N^* \mid K(w) \leq
	- s' .\log(\mu(w))\}$. Consider the following continued fraction
	$s$-gale $d_{s'} : \N^* \rightarrow [0,\infty)$
	\begin{align*}
		d_{s'}(v) =\frac{1}{\gamma^{s}(v)} \; \left( \sum_{w \in
			G_{s'} ; v \prefix w} \gamma^{s'}(w) + \sum_{w \in
			G_{s'} ; w \sqsubset v} \gamma^{s'}(w)
		\frac{\gamma(v)}{\gamma(w)}\right).
	\end{align*}
	
	We can see that for any $s'$, since $G_{s'}$ is computably enumerable,
	$d_{s'}$ is lower semicomputable.
	
	For each $w \in G$, let $\mathcal{E}(w) = b$. By definition $K(w) =
	K(b)$. Therefore, $K(b) \leq - s' \log(\mu(w)) $. From the Definition
	of $\mathcal{E}$ and Lemma \ref{lem:vishnoiMeasure}, we see that
	$\mu(b) \geq \frac{1}{16} \mu(w)$. The extra factor of $\frac{1}{4}$
	comes from the two additional bits that are padded in $b$. Therefore,
	we have that $K(b) \leq - s' \log(16.\mu(b)) \leq s'|b| - 4s' \leq s' |b|$.
	
	Define $B = \{x \in \{0,1\}^*$ $~|~$  $K(x) \leq s'|x|\}$. Hence if $w \in G$, then $\mathcal{E}(w) \in B$. 
	
	Now, $d_{s'}(\lambda) = \sum_{w \in G} \gamma^{s'}(w)$. Using Lemma
	\ref{lem:lebesgue_gauss}, we get that $d_{s'}(\lambda) \leq c \sum_{w
		\in G} \mu^{s'}(w)$ where $c = (\ln 2)^{-s'}$. Since $s' < 1$, we
	have that $c < \ln 2$.
	
	From Lemma \ref{lem:vishnoiMeasure}, it follows that
	$\mu^s(\mathcal{E}(w)) \geq \frac{1}{16} . \mu^s(w)$. Using the fact
	that $\mathcal{E}$ is a one to one encoding, we get that
	$d_{s'}(\lambda) \leq 16 \ln 2  \sum_{x \in B} 2^{-s'|x|}$.
	
	In the proof of the Kolmogorov complexity characterization of constructive dimension \cite{Mayordomo02}, Mayordomo shows that $\sum_{x \in B} 2^{-s'|x|} \leq 2^c$ for some constant
	$c$. The proof uses the fact that $|B^{=n}|
	\leq 2^{s'n - K(n)+c}$ along with the Kraft inequality $\sum_n 2^{-K(n)+c}
	\leq 2^c$.  Therefore $d_{s'}(\lambda)$ is finite.
	
	Note that for each $w \in G_{s'}$ we see that $d_{s'}(w) \geq
	(\frac{1}{\gamma(w)})^{s - s'}$.
	
	Now for each $i \in \N$, let $s_{i} = s(1- 2^{-i})$. Finally, consider
	the following continued fraction $s$-gale.
	$$d^*(v) = \sum\limits_{i = 1}^\infty2^{-i} d_{s_i}(v).$$
	
	Consider any $Y \in \N^\infty$ such that $\liminf\limits_{n
		\rightarrow \infty} \frac{K(Y \upharpoonleft n)}{- \log(\mu(Y
		\upharpoonleft n))}< s$. Let $i \in \N$ be the smallest number such
	that $\liminf\limits_{n \rightarrow \infty} \frac{K(Y \upharpoonleft
		n)}{- \log(\mu(Y \upharpoonleft n))} < s_i $.
	
	For each $w \in G_{s_i}$ we see that $d^*(w) \geq 2^{-i}
	(\frac{1}{\gamma(w)})^{s - s_i}$. We see that there are infinitely
	many prefixes $w$ of $Y$ for which $w \in G_{s_i}$. Hence it follows
	that $d^*$ succeeds on Y.
\end{proof}

We refer to Downey and Hirschfeldt's (Theorem 13.3.4 \cite{Downey10})
proof of the lower bound
on constructive dimension using Kolmogorov complexity. The proof
fundamentally uses properties of the universal lower semicomputable
super-martingale.

For any real having continued fraction dimension less than $s$, we
obtain a lower semicomputable binary $s$-gale that succeeds on it from
Lemma \ref{lem:untitledUsefuLemma}. We use the success of this binary $s$-gale along with the same properties of the
universal lower semicomputable super-martingale, to prove the following lemma.

\begin{lemma}\label{lem:liminflessthandim}
	For any $Y \in \N^\infty$ and any $s > \cdim_{CF}(Y)$, we have
	$\liminf\limits_{n \rightarrow \infty} \frac{K(Y \upharpoonleft n)}{-
		\log(\mu(Y \upharpoonleft n))} \leq s.$
\end{lemma}

\begin{proof} 
	Given any $s > dim_{CF}(Y)$, take an $s'$ such that $s > s' >
	dim_{CF}(Y)$. By definition, there exists a continued fraction
	$s'$-gale, say $d:\N^* \to \Rplus$ that succeeds on $Y$. For each
	prefix $v_i \prefix Y$, let $\mathcal{E}(v_i) = b_i$. From the
	definition of $\mathcal{E}$ and Lemma \ref{lem:vishnoiMeasure}, we get
	that $\mu(v) \geq \mu(b_i) \geq \frac{1}{16} \mu(v_i)$. The extra
	factor of $\frac{1}{4}$ comes from the two additional bits that gets
	added in $\mathcal{E}$.

	Using Lemma \ref{lem:untitledUsefuLemma}, we get a binary $s$-gale
	$h:\Sigma^* \to \Rplus$ such that forall $i$, $h(b_i) \geq c_s d(v_i)$
	for some constant $c_s$. Therefore $\limsup_{i \rightarrow \infty}
	h(b_i) = \infty$.
	
	Let $f$ be the universal lower semicomputable super-martingale
	\cite{LiVitanyi}. As $2^{(1-s) |b_i|} h(b_i)$ is a martingale, it
	follows that $f(b_i) \geq c_h 2^{(1-s)|b_i|} h(b_i)$ for some constant
	$c_h$ that depends only on $h$. Therefore, we have that $\limsup_{i
		\rightarrow \infty} \frac{f(b_i)}{2^{(1-s)|b_i|}} = \infty$.
	
	Let $|b_i| = n$. As noted in the proof of Theorem 13.3.4 in
	\cite{Downey10}, $f(b_i) = 2^{n - KM(b_i) \pm O(1) }$ where $K(b_i)
	\leq KM(b_i) + O(\log n)$. Hence it follows that $\limsup_{i
		\rightarrow \infty} 2^{sn - K(b_i) + O(\log n)} = \infty$. Therefore
	for infinitely many $i\in\N$, we have $K(b_i) < sn + O(\log n)$.
	
	By definition $K(v_i) = K(b_i)$. Also, we have that $|b_i| =
	-\log(\mu(b_i)) \leq -\log(\mu(v_i)) + 4$. Therefore for infinitely
	many $i\in\N$, we have $K(v_i) < -s\log(\mu(v_i)) + 4s + O(\log (-
	\log (\mu(v_i))))$. From this, it follows that $\liminf\limits_{i
		\rightarrow \infty} \frac{K(v_i)}{-\log(\mu(v_i))} \leq s$.
\end{proof}

Therefore, we have the following Kolmogorov complexity based
characterization of effective continued fraction dimension.
\begin{theorem} \label{thm:dimcfEqualsLiminf}
	For any $Y \in \N^\infty$,
	$$ dim_{CF}(Y) = \liminf\limits_{n \rightarrow \infty} \frac{K(Y
		\upharpoonleft n)}{- \log(\mu(Y \upharpoonleft n))}.$$
\end{theorem}
\begin{proof}
	For any $Y \in \N^\infty$, let $s^* = \liminf\limits_{n \rightarrow
		\infty} \frac{K(Y \upharpoonleft n)}{- \log(\mu(Y \upharpoonleft
		n))}$.
	
	For any $s > s^*$, from Lemma \ref{lem:dimlessthanliminf}, it follows
	that there exists a lower semicomputable $s$-gale $\D$ that succeeds on
	$Y$. Hence $\dim_{CF}(Y) \leq s^*$.
	
	For any $s > \dim_{CF}(Y)$, from Lemma \ref{lem:liminflessthandim}, we
	have that $s^* \leq s$. Therefore, we have $s^* \leq \dim_{CF}(Y)$.
\end{proof}

\subsection{Optimal gales and effective continued fraction dimension of a set}\label{sec:optSgale}

Lutz \cite{Lutz2003} utilizes the notion of the optimal constructive
subprobability supermeasure $\mathbf{M}$ on the Cantor space \cite{ZvonkinLevin70} to provide a
notion of an optimal constructive supergale.

We note that using Theorem \ref{thm:dimcfEqualsLiminf}, the gale that
we obtain from Lemma \ref{lem:dimlessthanliminf} leads to an analogous
notion in the continued fraction setting. We call the continued
fraction $s$-gale $d^*$ thus obtained as the \emph{optimal
	lower semicomputable continued fraction $s$-gale}.

\begin{lemma} \label{lem:almostOptimalSgale}
	For any $s>0$, there exists a lower semicomputable continued fraction
	$s$-gale $d^* : \N^* \rightarrow \lcro{0}{\infty}$ such that
	for all $Y \in \N^\infty$ with $\cdim_{CF}(Y)< s$, $d^*$
	succeeds on $Y$.
\end{lemma}
\begin{proof}
	For all $Y \in \N^*$ such that $\cdim_{CF}(Y) < s$, from Theorem
	\ref{thm:dimcfEqualsLiminf}, it follows that $\liminf\limits_{n
		\rightarrow \infty} \frac{K(Y \upharpoonleft n)}{- \log(\mu(Y
		\upharpoonleft n))}< s$. Now applying Lemma
	\ref{lem:dimlessthanliminf}, we see that the given lemma holds.
\end{proof}

Lutz (Theorem 4.1 in \cite{Lutz2003}) shows that the effective
dimension of a set is precisely the supremum of effective dimensions
of individual elements in the set, that is for all $X \subseteq
[0,1]$, $\cdim(X) = \sup_{S \in X} \cdim(S) $. Using the notion of the
optimal lower semicomputable continued fraction $s$-gale from Lemma
\ref{lem:almostOptimalSgale}, we extend this result to continued
fraction dimension.


\begin{theorem} \label{thm:SetToPoint}
	For all $\mathcal{F} \subseteq [0,1]$, $\cdim_{CF}(\mathcal{F}) =
	\sup_{Y \in \mathcal{F}} \cdim_{CF}(Y)$.
\end{theorem}
\begin{proof} 
	For any $s > \cdim_{CF}(\mathcal{F})$, for all $Y \in \mathcal{F}$
	there exists a lower semicomputable continued fraction $s$-gale that
	succeeds on $Y$. Thus we have $\sup_{Y \in \mathcal{F}} \cdim_{CF}(Y)
	\leq s$.
	
	Take any any $s > \sup_{Y \in \mathcal{F}} \cdim_{CF}(Y)$. It follows
	that for all $Y \in \mathcal{F}$, $\cdim_{CF}(Y) < s$. Therefore from
	Lemma \ref{lem:almostOptimalSgale}, we have that there exists a lower
	semicomputable continued fraction $s$-gale $d^* : \N^* \rightarrow
	\lcro{0}{\infty}$ that succeeds on all $Y \in \mathcal{F}$. Therefore,
	$\cdim_{CF}(\mathcal{F}) \leq s$.
\end{proof}

\section{Reals with unequal Effective Dimension and Effective
	Continued Fraction Dimension}
\label{sec:counterexample}

In this section, we show that for any set of reals
$\mathcal{F}\subseteq [0,1]$, the effective Hausdorff effective dimension of $\mathcal{F}$ 
cannot exceed its effective  continued fraction
dimension. We show that this cannot be improved to an equality. Hence,
this inequality is strict in general. We show this by proving the
existence of a real such that its effective continued fraction
dimension is strictly greater its effective dimension.
\subsection{Effective Hausdorff dimension is at most the effective
	continued fraction dimension}
\begin{theorem} \label{lem:cdimleqcfdim}
	For any $\mathcal{F} \subseteq [0,1]$, $\cdim(\mathcal{F}) \leq
	\cdim_{CF}(\mathcal{F})$.
\end{theorem}

\begin{proof}
	Let $s > s' >\cdim_{CF}(\mathcal{F})$. By definition, there exists a
	lower semicomputable continued fraction $s'$-gale $d: \N^* \to
	\lcro{0}{\infty}$ such that $\mathcal{F} \subseteq S^\infty[d]$.
	

	Take any $Y \in S^\infty[d]$. Let $X \in \Sigma^\infty$ be the
	corresponding binary representation of $Y$. By definition, for
	any $m \in \N$, there exists an $n \in \N$ such that $d(Y
	\restr n) > m$. Let $v = Y \restr n$.
	
	Using Lemma \ref{lem:binary_covering_cf}, we get two binary
	cylinders $w_1$ and $w_2$ such that $C_v \subseteq C_{w_1}
	\cup C_{w_2}$ such that $\mu(w_1) = \mu(w_2) \leq 2 \mu(v)$. We have that since $v \prefix Y$, $w1 \prefix X$ or $w2 \prefix X$. Without loss of generality assume that $w_1 \prefix X$.
	
	From Lemma \ref{lem:untitledUsefuLemma}, we obtain a lower
	semicomputable $s$-gale $h$ such that $h(w_1) \geq c_s . d(v)
	\geq c_s. m$ for some positive constant $c_s$.
	
	Since $m$ is arbitrary, we see that $h$ succeeds on $X$.
\end{proof}

\subsection{Reals with unequal effective Hausdorff and effective
	continued fraction dimensions} 
We now provide the main construction of the paper, utilizing the
results we show in previous sections.

We first require some technical lemmas about the estimation of
Lebesgue measure of a continued fraction cylinder in terms of digits of the continued fraction. Some of
the bounds derived in this section may be of independent interest.

In combinatorial arguments, the Gauss measure is often difficult to
deal with directly, and it is convenient to use the Lebesgue measure,
and derive inequalities.

The following equation, Proposition 1.2.7 in Kraaikamp and Iosifescu
\cite{KraaikampIosifescu}, is extremely useful in deriving an estimate for the
Lebesgue measure of continued fraction cylinders. We derive consequences
of this Lemma below, and these are crucial in estimating the dimension
of the continued fraction we construct in Section
\ref{sec:counterexample}. Note that the bounds for Gauss measure are
not simple to derive directly.

\begin{lemma} [Kraaikamp, Iosifescu
	\cite{KraaikampIosifescu}]\label{lem:Kraaikamp} 
	For any $v = [a_1, \dots a_n]$ and $i \in
	\N$, $$\frac{\mu([v,i])}{\mu([v])} = \frac{s_n + 1}{(s_n + i)(s_n + i
		+ 1)}$$ where $s_n = [a_n , \dots a_1]$ is the rational
	corresponding to the reverse of string $v$.
\end{lemma}

The lemma given above gives the following bounds on the Lebesgue measure of
a continued fraction cylinder in terms of the digits of the continued
fraction.

\begin{lemma} \label{lem:KraaikampExtended}
	For any $v = [a_1, \dots a_k ] \in \N^k$ we have 
	$$\prod\limits_{i =
		1}^k \frac{1}{(a_i + 1)(a_i + 2)}
	\ \leq\ \mu(v)\ \leq\ \prod\limits_{i = 1}^k \frac{2}{a_i \; (a_i +
		1)}
	$$
\end{lemma}
\begin{proof} 
	If $k = 1$, then $\mu(v) = \frac{1}{a_1 (a_1+1)}$, therefore the lemma
	holds. For $k>1$, let $v = [a_1, \dots a_{k-1}]$ and consider $[v,a_k]$.
	From Lemma \ref{lem:Kraaikamp}, it follows that $$\mu([v,a_k]) =
	\frac{s_{k-1} + 1}{(s_{k-1} + a_k ) (s_{k-1} + a_k + 1)} \;\;
	\mu(v),$$ where $s_{k-1} = [a_{k-1}, \dots a_1]$. Since $s_{k-1} \in
	[0,1]$, it follows that
	
	$$\frac{1}{(a_k +1) (a_{k} + 2)} \;\mu(v)\leq \mu([v,a_k]) \leq
	\frac{2}{(a_k ) (a_{k} + 1)} \; \mu(v)$$
	
	Therefore, by induction on $k$, the lemma holds.
\end{proof}


\begin{lemma}\label{lem:KraaikampExtended2}
	Let $v = [a_1 \dots a_k ] \in \N^*$. Then for any $a,b \in \N$ such
	that $b>a$,
	\begin{align*}\mu\left(\;\bigcup\limits_{i = a}^b \;[v, i]\;\right) \leq \frac{2}{a}\; \prod\limits_{i = 1}^k \frac{2}{a_i \; (a_i + 1)}
	\end{align*}
\end{lemma}

\begin{proof} 
	From Lemma \ref{lem:Kraaikamp}, it follows that
	$$\frac{\mu([v,i])}{\mu(v)} = \frac{s_k + 1}{(s_k + i)(s_k + i + 1)} $$
	where $s_k \in [0,1]$. Therefore 
	$$ \mu\Big(\;\bigcup\limits_{i = a}^b \;[v, i]\;\Big) \leq
	\sum\limits_{i=a}^{b}\frac{2}{i(i+1)} \mu(v) \leq \frac{2}{a} \;
	\mu(v).$$ From Lemma \ref{lem:KraaikampExtended}, $\mu(v) \leq
	\prod\limits_{i = 1}^k \frac{2}{a_i \; (a_i + 1)}$. Therefore this
	lemma holds.
\end{proof} 

The following lemma is a direct constructive extension of the proof by Lutz \cite{Lutz03}. Using this technique, we convert a set of computably
enumerable prefix free binary covers into a lower semicomputable binary $s$-gale.

\begin{lemma} [Lutz \cite{Lutz03}]\label{lem:scoverTosgale}
	For all $n \in \N$, and $\F \subseteq [0,1]$, if there is a computably
	enumerable prefix free binary cover $\{B^n_i\}$ of $\F$, such that
	$\sum_i|B_i^n|^s < 2^{-n}$, then there exists a lower semicomputable binary
	$s$-gale that succeeds on $\F$.
\end{lemma}

\begin{proof} 
	For an $n \in N$, define an $s$-gale $d_n: \{0,1\}^* \rightarrow [0,\infty)$ as
	follows:
	\medskip
	
	If $\exists \; v \prefix w$ such that $v \in \{B^i_n\}$, then 
	$$d_n(w) = \frac{\mu^s(v)}{\mu^s(w)} \frac{\mu(w)}{\mu(v)}$$
	
	Otherwise,
	$$d_n(w) = \frac{1}{\mu^s(w)}{ \sum_{v \in A_r \; ; \; w \sqsubset v } \mu^s(v)}$$
	
	Finally define $$d(w) = \sum_{n=0}^{\infty}2^n d_{2n}(w).$$
	
	It is straightforward to verify that $d_n$ is an $s$-gale. It follows that $d_n(\lambda) \leq  2^{-n}$, hence $d(\lambda) \leq 1 $.  Since $\{B^i_n\}$ is computably enumerable, $d$ is lower semicomputable. Now, for all $ X \in F$ and  $n \in \N$, there exists a prefix $w \prefix X$ such that $w \in \{B^i_n\}$. From the definition of $d$, we can see that $\forall w \in \{B^i_n\}$, $d_n(w) = 1$.
	Therefore, $d(w) \geq 2^n d_{2n}(w) = 2^n.$ Since $n$ is arbitrary, the $s$-gale d succeeds on $F$. 
\end{proof} 

We now proceed to show the construction of the set $\F$. The definition uses a parameter $\mathbf{s}$. We later go on to show that for all such $\F$, there exists an element $Y \in \mathcal{F}$ such that the
effective continued fraction dimension of $Y$ is greater than $0.5$. We also go on to show that $cdim(F) \leq \mathbf{s}$.

We first provide the stage-wise construction of a set $\mathcal{F}_k \subseteq [0,1]$, such that for each $k \in \N$ $\mathcal{F}_{k+1} \subseteq \mathcal{F}_k$. We then define the set $\mathcal{F}$ using an infinite intersection of the sets $\mathcal{F}_k$. 

\begin{definition} \label{def:mathcalF}
	Let $0 < \mathbf{s} < 0.5$. Define $a_1 = 1$. For any
	$k \in \mathbb{N}$, such that $k>1$, define $a_k$ inductively as:
	$$
	a_k = 2\bigg(k  \prod\limits_{i = 1}^{k-1} 100
	a_i\bigg)^{1/\mathbf{s}}.
	$$

	For any $k \in \mathbb{N}$, define $b_k = 50.a_k$. Take $\mathcal{F}_0
	= \lambda$.
	
	Let $\mathcal{F}_k = \{[v_1 \dots v_k] \in \N^k \text{ such that } v_i
	\in [a_i , b_i] \text{ for } i \in 1 \text{ to } k\}$. Finally define
	$$\mathcal{F} = \bigcap\limits_{k = 1}^\infty \mathcal{F}_k.$$
\end{definition}


We use the bounds obtained from Lemma \ref{lem:Kraaikamp}, along with
basic properties of harmonic numbers to prove the following property
of measures of continued fraction sub cylinders.

\begin{lemma}
	\label{lem:sMeasureIncreases}
	For any $x \in \N^*$ , $s \leq 0.5$ and $a_k,b_k \in N$ such that $b_k
	= 50.a_k$, \\$\sum\limits_{i = a_k}^{b_k} \gamma^s([x,i]) > c
	\gamma^s([x])$ for some $c > 1$.
\end{lemma}	

\begin{proof}
	From Lemma \ref{lem:Kraaikamp}, it follows that 
	\begin{align*}
		\sum\limits_{a_n}^{b_k} \frac{\mu^s([x,i])}{\mu^s([x])} &=  \sum\limits_{a_n}^{b_k} \bigg(\frac{s_k + 1}{(s_k+i)(s_k + i + 1)}\bigg)^s\\	
		& \geq \sum\limits_{a_k}^{b_k} \bigg(\frac{1}{(i+1)(i+2)}\bigg)^s
	\end{align*}
	
	The second inequality follows from the fact that $s_k \in [0,1]$. 
	
	Using Lemma \ref{lem:lebesgue_gauss}, we get that
	\begin{align*}
		\sum\limits_{a_n}^{b_k} \frac{\gamma^s([x,i])}{\gamma^s([x])} 	
		& \geq \sum\limits_{a_k}^{b_k} \bigg( \frac{1}{2(i+1)(i+2)}\bigg)^s
	\end{align*}
	
	Putting $b_k = 50 a_k$ and $s \leq 0.5$, we get
	\begin{align*}
		\sum\limits_{a_n}^{b_k}
		\frac{\gamma^s([x,i])}{\gamma^s([x])}  & \geq
		\frac{1}{2} \sum_{a_k }^{50 a_k } \frac{1}{i + 2} \\ &
		= 0.5(H(50 a_k + 2) - H(a_k + 1)). 
	\end{align*} 
	
	where $H_n$ is the $n^{th}$ Harmonic number. From the fact that $\ln n \leq H_n \leq \ln n + 1$, we have 
	\begin{align*}
		H(50 a_k + 2) - H(a_k + 1)  & \geq \ln (50.a_k) - \ln(2.a_k) - 1\\
		& = \ln(25) - 1.
	\end{align*} 
	
	The lemma holds as $0.5(\ln 25 - 1)$ is greater than $ 1$.
\end{proof} 

Using the bound derived above, we show that for $s=0.5$, the optimal $s$-gale $d^*$
formulated in Section \ref{sec:optSgale} does not succeed on some
sequence in $Y \in \mathcal{F}$. Using this we establish that $\cdim_{CF}(Y) \geq 0.5$.

\begin{lemma} \label{lem:dimYinCFgt0.5}
	There exists a $Y \in \mathcal{F}$ such that $\cdim_{CF}(Y) \geq 0.5$.
\end{lemma}

\begin{proof} Let $s = 0.5$.
	Consider the continued fraction $s$-gale $d^*$ from Lemma
	\ref{lem:almostOptimalSgale}. It follows that for all $Y \in \N^*$, if
	$d^*$ does not succeed on any $Y \in \N^*$, then
	$\cdim_{CF}(Y) \geq s$.
	
	Consider any $v \in \N^*$, let $rank(v) = k$. 
	From lemma
	\ref{lem:sMeasureIncreases}, we have that for some $c >1$, $\sum\limits_{i = a_k}^{b_k} \gamma^s([v,i]) > c . \gamma^s([v]).$
	
	Now if $\forall i \in [a_k,b_k]$, $d^*([v,i]) \geq \frac{1}{c}.d^*(v)$, from the $s$-gale condition it follows that $d^*(v) \gamma^s(v) \geq \frac{d^*(v)}{c} \sum\limits_{i = a_k}^{b_k} \gamma^s([v,i]) > d^*(v) \gamma^s(v)$, which is a contradiction.
	
	Therefore,
	it follows that for all $v \in \N^*$, there exists some $n_v
	\in [a_k,b_k]$ such that $d^*([v,i]) < \frac{1}{c} .
	d^*([v])$.
	
	Let $v_0 = \lambda$, for each $i \in \N$, we define $v_i = [v_{i-1},
	n_{v_{i-1}}]$. Take $Y = \cap_{i = 1}^{\infty} C_{v_i}$, it follows
	that $Y \in F$. Taking $d^*(\lambda) = k$ we get $d^*(Y
	\restr n) < \frac{k}{c^n}$. Therefore the continued fraction $s$-gale
	$d^*$ does not succeed on $Y$. Hence $\cdim_{CF}(Y) \geq 0.5$.
\end{proof}

From this, it follows that the constructive dimension of the entire set $\mathcal{F}$ is also greater than or equal to 0.5. 
\begin{lemma}
	$\cdim_{CF}(\mathcal{F}) \geq 0.5$.
\end{lemma}
\begin{proof}
	From Theorem \ref{thm:SetToPoint}, we get that
	$\cdim_{CF}(\mathcal{F}) = \sup_{Y \in \mathcal{F}} \cdim_{CF}(Y)$.
	From Lemma \ref{lem:dimYinCFgt0.5}, it follows that there exists a $Y
	\in F$ such that $\cdim_{CF}(Y) \geq 0.5$. Combining these two, we get
	that $\cdim_{CF}(\mathcal{F}) \geq 0.5$.
\end{proof}

Now we show that the effective Hausdorff dimension of all points in the set
$\mathcal{F}$ is less than $\mathbf{s}$. Using ideas from \cite{Albeverio2017}, we devise a set of covers $S_k$ for $\mathcal{F}$, by combining adjacent continued fraction cylinders into a single cover. 

Using the bounds derived on Lebesgue measure of continued fraction cylinders, we show that for the set of covers $S_k$ for $\mathcal{F}$,
the $\mathbf{s}$-dimensional Hausdorff measure shrinks
arbitrarily small.

\begin{lemma} \label{lem:smeasureF0}
	For $k \in \N$, let $S_k = \{\ \bigcup\limits_{i = a_k}^{b_k} \; [v,i]
	\text{ for } v \in \mathcal{F}_{k-1}\}$. Then, $ \sum\limits_{S \in
		S_k} \mu^{\mathbf{s}}(S) \leq 1/k$.
\end{lemma}
\begin{proof}
	The largest element in $ S_k$ is $I = [(a_1 \dots a_{k-1}, a_k),[(a_1
	\dots a_{k-1}, b_k)]$. The number of elements in $S_k$ equals
	$\prod_{i = 1}^{k-1} (b_{i} - a_{i})$. Additionally, we have $b_i =
	50 a_i$ for all $i \in \N$. Therefore,
	$$\sum\limits_{S \in S_k} \mu^{\mathbf{s}}(S) \leq  \mu^{\mathbf{s}}(I)  \prod\limits_{i =
		1}^{k-1} 50 a_{i}.$$
	
	From Lemma \ref{lem:KraaikampExtended2}, it follows that $\mu(I) \leq \frac{2}{a_k}\; \prod\limits_{i = 1}^{k-1} \frac{2}{a_i \; (a_i + 1)}$. Therefore,
	\begin{align*}
		\sum\limits_{S \in S_k} \mu^{\mathbf{s}}(S) &\leq \Big( \frac{2}{a_k}\; \prod\limits_{i = 1}^{k-1} \frac{2}{a_i^2 \;} \Big)^{\mathbf{s}}  \;\; \prod\limits_{i = 1}^{k-1}  (50  a_{i})\\
		& \leq \frac{2^{\mathbf{s}}}{a_k^{\mathbf{s}}}  \;\; \prod\limits_{i =
			1}^{k-1} 100 a_i
	\end{align*}
	Since $a_k = 2 (k \prod\limits_{i = 1}^{k-1} 100 a_i)^{1/{\mathbf{s}}} $,
	this value is less than $1/k$.
\end{proof}

To show that the constructive dimension of $\mathcal{F}$ is less than $\mathbf{s}$, we construct a lower
semicomputable binary $\mathbf{s}$-gale that succeeds on $\mathcal{F}$.
Using standard techniques, we first convert the covers obtained in Lemma \ref{lem:smeasureF0} to a set of binary covers of $\mathcal{F}$. Finally applying Lemma \ref{lem:scoverTosgale}, we convert the binary covers into a semicomputable $\mathbf{s}$-gale that succeeds on $\mathcal{F}$.

\begin{lemma} \label{lem:cdimFlessthns}
	$\cdim(\mathcal{F}) \leq \mathbf{s}.$
\end{lemma}
\begin{proof}
	Given $k\in\N$, from Lemma \ref{lem:smeasureF0}, we have that for $S_k
	= \{\ \bigcup_{i = a_k}^{b_k} \; [v,i] \text{ for } v \in
	\mathcal{F}_{k-1}\}$, $ \sum_{S \in S_k} \mu^{\mathbf{s}}(S) \leq
	1/k$. For each $S \in S_k$, using Lemma
	$\ref{lem:binary_covering_cf}$, we get that for the two smallest
	consecutive binary cylinders say $b_1(S)$ and $b_2(S)$ that cover $S$,
	we have that $\mu(b_1) = \mu(b_2) \le 2 \mu(C)$.
	
	Hence the set $B_k = \{\{b_1(S)\} \cup \{b_2(S)\} \text{ such that } S \in S_k\}$ forms a binary cover of $S_k$. 
	Also from Lemma \ref{lem:binary_covering_cf}, we have that $\sum_{b
		\in B_k} \mu^{\mathbf{s}}(b) \leq 2^{1+{\mathbf{s}}} \sum_{S \in
		S_k} \mu^{\mathbf{s}}(S) \leq 2^{1+{\mathbf{s}}}/k$.

	Note that the set $S_k$ is computable as $a_k$ and $b_k$ are
	computable for all $k$. Given any interval $S$, $b_1(S)$ and $b_2(S)$ are also computable. Hence the set $B_k$ is computable.
	
	Since $B_k$ is a finite set, we can remove all $v \in B_k$ such that $u \sqsubset v$ for some $u \in B_k$, to make $B_k$ prefix free.
	
	For an $n \in \N$, taking $k = \lceil 2^{1 + \mathbf{s}} .
	2^n\rceil$, the set $B_{k}$ forms a computably enumerable prefix free binary cover of
	$\mathcal{F}$ such that $\sum_{b \in B_k} \mu^{\mathbf{s}}(b) \leq
	2^{-n}$.
	
	Applying Lemma \ref{lem:scoverTosgale}, we get that there exists a
	lower semicomputable $\mathbf{s}$-gale that succeeds on $\mathcal{F}$.
	Hence the lemma holds.
\end{proof}

We sum up the results from Lemma \ref{lem:dimYinCFgt0.5} and Lemma
\ref{lem:cdimFlessthns} into the following theorem.

\begin{theorem}
	Given any $0 <\varepsilon < 0.5$, there exists a $Y \in \N^\infty$ such
	that $\cdim_{CF} (Y) \geq 0.5$ and $\cdim(Y) \leq \varepsilon$.
\end{theorem}

\begin{proof}
	Given $0 < \varepsilon < 0.5$, taking $\mathbf{s} = \varepsilon$,
	construct the set $\mathcal{F}$ given in Definition
	\ref{def:mathcalF}.

	From Lemma \ref{lem:dimYinCFgt0.5}, it follows that there exists a $Y
	\in \mathcal{F}$ such that $\cdim_{CF}(Y) \geq 0.5 $.
	
	From Lemma \ref{lem:cdimFlessthns}, it follows that for all $X \in \F$,
	$cdim(X) \leq \varepsilon$. Hence $cdim(Y) \leq \varepsilon$.

\end{proof}

\bibliographystyle{plain} 
\bibliography{main.new}

\end{document}

%% file: commands.tex
\usepackage{amsfonts}
\usepackage{amsmath} 
\usepackage{amssymb}
\usepackage{enumerate}

\newcommand{\PROOF}{\begin{proof}}

\newcommand{\QED}{\end{proof}}

\newcommand{\prefix}{\sqsubseteq}

\newcommand{\restr}{\upharpoonright}

\newcommand{\N}{\mathbb{N}}

\newcommand{\cdim}{\mathrm{cdim}}

\renewcommand{\dim}{{\mathrm{dim}}}

\newcommand{\D}{{\mathcal{D}}}

\newcommand{\F}{{\mathcal{F}}}

\newcommand{\lcro}[2]{[#1,#2)}

\newtheorem{theorem}{Theorem} 
\newtheorem{lemma}{Lemma}

\newtheorem{definition}{Definition}
\newtheorem*{theorem*}{Theorem}